\documentclass[11pt,a4paper]{article}

\usepackage[pagebackref=false,citecolor=newblue,colorlinks=true,linkcolor=newblue,bookmarks=false]{hyperref}
\usepackage{epsfig,amssymb,amsfonts,amsmath,amsthm}
\usepackage[margin=1in]{geometry}

\usepackage{booktabs} 

\usepackage[ruled]{algorithm2e} 

\usepackage{subcaption}

\usepackage{tikz}
\usepackage{setspace}

\definecolor{ocre}{rgb}{0.72,0,0} 
\definecolor{newblue}{rgb}{0.2,0.2,0.6} 
\definecolor{babyblueeyes}{rgb}{0.63, 0.79, 0.95}

\definecolor{newgreen}{rgb}{0.53,0.66,0.42} 

\newcommand{\Dim}{\operatorname{dim}}
\newcommand{\calL}{\mathcal{L}}
\newcommand{\mat}[1]{\mathbf{#1}}

\numberwithin{equation}{section}

\newcommand{\R}{\mathbb{R}}

\newcommand{\poly}{\operatorname{poly}}

\newcommand{\vol}{\operatorname{vol}}

\newcommand{\Span}{\operatorname{span}}

\renewcommand{\leq}{\leqslant}
\renewcommand{\geq}{\geqslant}
\renewcommand{\le}{\leqslant}
\renewcommand{\ge}{\geqslant}
\newcommand{\thmref}[1]{Theorem~\ref{thm:#1}}

\newcommand{\lemref}[1]{Lemma~\ref{lem:#1}}

\newcommand{\figref}[1]{Figure~\ref{fig:#1}}

\newcommand{\secref}[1]{Section~\ref{sec:#1}}

\newcommand{\eq}[1]{\eqref{eq:#1}}

\newcommand{\Pro}[1]{\mathbf{P} \left[\,#1\,\right]}

\newcommand{\Ex}[1]{\mathbf{E} \left[\,#1\,\right]}

\renewcommand{\tilde}{\widetilde}
\renewcommand{\epsilon}{\varepsilon}

\newcommand{\aseq}{\{A_i\}_{i=1}^k}

\newcommand{\Id}{\mathbf{I}}

\newtheorem{thm}{Theorem}[section]  
\newtheorem{theorem}[thm]{Theorem}

\newtheorem{lemma}[thm]{Lemma}

\newtheorem*{rem*}{Remark}

\renewcommand{\tilde}{\widetilde}
\usepackage{todonotes}

\begin{document}

\title{{\bf Distributed Graph Clustering and Sparsification}}

\author{He Sun\\
University of Edinburgh\\ Edinburgh, UK\\
\texttt{h.sun@ed.ac.uk}
\and
Luca Zanetti\\
University of Cambridge\\
Cambridge, UK\\
\texttt{luca.zanetti@cl.cam.ac.uk}
}

\date{}

%
%
%
%

\maketitle

\begin{abstract}
Graph clustering is a fundamental computational problem  with a number of applications in algorithm design, machine learning, data mining, and analysis of social networks. 
Over the past decades,  researchers have proposed  a number of algorithmic design methods for graph clustering. Most of these methods, however, are based on complicated spectral techniques or convex optimisation, and cannot be directly applied  for clustering many networks that occur in practice, whose information is often collected on different sites.  Designing a simple and distributed clustering algorithm is of great 
interest, and has wide applications for processing big datasets.

In this paper we present a simple and distributed algorithm for graph clustering: for a wide class of graphs that are characterised by a strong cluster-structure, our algorithm finishes in a poly-logarithmic number of rounds, and recovers a partition of the graph
close to optimal. One of the main components behind our algorithm is a sampling scheme that, given a dense graph as input, produces a sparse subgraph that provably preserves the cluster-structure of the input. Compared with previous sparsification algorithms that require Laplacian solvers or involve combinatorial constructions, this component is easy to implement in a distributed way and runs fast in practice.

\vspace{0.5cm}

\textbf{Keywords:} graph clustering, graph sparsification, distributed computing

\end{abstract}

\thispagestyle{empty}

\setcounter{page}{0}

\newpage

\section{Introduction}
Analysis of large-scale networks has brought significant advances to our understanding of complex systems. One of the most relevant features of  the networks occurring in practice is their structure of clusters, i.e., an organisation of nodes into clusters such that nodes within the same cluster are highly connected in contrast to nodes from different clusters. Graph clustering is an important research topic
in many disciplines, including computer science,  biology, and sociology. For instance, graph clustering is widely used in finding communities in social networks,  webpages dealing with similar topics, and proteins having the same specific function within the cell in protein-protein interaction networks~\cite{fortunatoPR}. However, 
despite extensive studies on efficient methods for graph clustering,  many approximation algorithms for this problem requires advanced algorithm design techniques, e.g., spectral methods, or convex optimisation, which make the algorithms difficult to be implemented in the distributed setting, where graphs are allocated in sites which are physically remote. Designing a simple  and distributed algorithm is of important interest in practice, and has received considerable attention in recent years~\cite{hui2007distributed, CSWZ16,yang2015divide}.


\subsection{Structure of Clusters}  Let  $G=(V,E,w)$ be an undirected graph with $n$ nodes and weight function $w: V\times V\rightarrow\mathbb{R}_{\geq 0}$.
 For any set $S$,  let the conductance of $S$ be
\[
\phi_G(S)\triangleq\frac{w(S,V\setminus S)}{\vol(S)},
\]
where $w(S, V\setminus S)\triangleq \sum_{u\in S} \sum_{v\in V\setminus S} w(u,v)$ is the total weight of edges 
 between $S$ and $V\setminus S$, and $\vol(S)\triangleq \sum_{u\in S} \sum_{u\sim v} w(u,v)$ is the volume of $S$.  Intuitively, nodes in $S$ form a cluster if $\phi_G(S)$ is small, i.e. there are fewer connections between  the nodes of $S$ to the nodes in $V\setminus S$.  We call subsets of nodes~(i.e.~\emph{clusters}) $A_1,\ldots, A_k$ a \emph{$k$-way partition} of $G$ if $A_i\cap A_j=\emptyset$ for different $i$ and $j$, and $\bigcup_{i=1}^k A_i=V$.   Moreover, we define the  \emph{$k$-way expansion constant} by
\[
\rho(k)\triangleq\min\limits_{\mathrm{partition\ } A_1,
\ldots, A_k }\max_{1\leq i\leq k} \phi_G(A_i).
\]
Computing the exact value of $\rho(k)$ is $\mathsf{coNP}$-hard, and a sequence of results show that $\rho(k)$ can be approximated by algebraic quantities relating to the matrices  of $G$.   For instance,
Lee et al.~\cite{lot} shows the following high-order Cheeger inequality:
\begin{equation}\label{eq:cheeger}
\frac{\lambda_k}{2}\leq \rho(k)\leq O\left(k^2\right)\sqrt{\lambda_k},
\end{equation}
where $0=\lambda_1\leq \cdots\leq \lambda_n\leq 2$ are the eigenvalues of the normalised Laplacian matrix of  $G$.
Based on \eq{cheeger}, we know that  a large gap between $\lambda_{k+1}$ and $\rho(k)$ guarantees (i) existence of a $k$-way partition $S_1,\ldots S_k$ with bounded $\phi_G(S_i)\leq \rho(k)$, and (ii) any $(k+1)$-way partition $A_1,\dots,A_{k+1}$ of $G$ contains a subset $A_i$ with significantly higher conductance $\rho(k+1)\geq \lambda_{k+1}/2$ compared with $\rho(k)$.  Peng et al.~\cite{peng15} formalises these observations by defining the parameter
\[
\Upsilon_G(k)\triangleq\frac{\lambda_{k+1}}{\rho(k)},
\]
and shows that a suitable lower bound on the gap for $\Upsilon_G(k)$
implies that $G$ has $k$ well-defined clusters.


\subsection{Our Results}

The first result of our paper is a  simple algorithm that, given as input any graph $G$ with a well-defined cluster-structure, produces a sparse subgraph $H$ of $G$ that preserves the same cluster-structure of $G$, but has an almost-linear number of edges. The result is summarised as follows:
  
\begin{theorem}\label{main11}
There exists an algorithm that, receiving as input a graph $G=(V,E,w)$ with $k$ clusters and a parameter $\tau$ such that $\tau \ge C/\lambda_{k+1}$ for a large enough constant $C>0$, with probability greater than $0.99$, computes a sparsifier $H=(V,F\subset E,\tilde{w})$ with  $|F|=O(n\tau\cdot\log n)$ edges such that the following holds:
\begin{enumerate}
\item $\Upsilon_H(k)=\Omega(\Upsilon_G(k)/k)$;
\item It holds for any $1\leq i\leq k$ that $\phi_H(S_i)=O(k\cdot \phi_G(S_i))$.
\end{enumerate}
Moreover, this algorithm can be implemented in $O(1)$ rounds in the distributed setting, and the total information exchanged among all nodes is $O(n\tau \cdot\log n)$ words. \end{theorem}

The first property $\Upsilon_H=\Omega(\Upsilon_G(k)/k)$ of the theorem ensures that the gap in  $H$ is preserved as long as $\Upsilon_G(k)\gg k$. The second property further shows that the conductance of 
each optimal cluster $S_i$ in $G$  
is approximately preserved in $H$ up to a factor of $k$, therefore  $S_i$ is  a low-conductance subset in $H$ as well. 
We remark that these $k$ clusters $S_1,\ldots, S_k$ might not form an optimal clustering in $H$ anymore.  However, this is not an issue, since every cluster with low conductance in $H$ has high overlap with its optimal correspondence. Hence, any algorithm that recovers a clustering close to the optimal one in $H$  will recover a clustering close to the optimal one in $G$.
Moreover, since $\lambda_{k+1}$ represents the inner-connectivity of the clusters, it is usually quite high: for most interesting cases we can assume $\lambda_{k+1} = \Omega\left(1/\poly(\log{n})\right)$. Indeed, the experiments described in \secref{exp} show $\tau \le 2$ works for all the tested datasets. 

The second result of the paper is a distributed algorithm to  partition a graph $G$ that possesses a cluster-structure with clusters of balanced size. The result is summarised  as follows:
\begin{theorem}\label{thm:cluster}
There is a distributed algorithm that, given as input a  graph $G=(V,E,w)$ with $n$ nodes, $m$ edges, and $k$ optimal clusters $S_1,\ldots, S_k$ with $\vol(S_i) \ge \beta \vol(V)$ for any $1\le i \le k$ and
 \begin{equation}
\label{eq:gapabs}
\Upsilon_G(k) = \omega\left( k^4 \log^2\frac{1}{\beta} + \log{n} \right),
\end{equation} 
  finishes in
\[
T \triangleq \Theta\left(\frac{\log n}{\lambda_{k+1}}\right)
\]
rounds, and with probability greater than $0.99$  the following statements hold: 
\begin{enumerate}
\item Each node $v$ receives a label $\ell_v$ such that the total volume of misclassified nodes is $o(\vol(V))$, i.e., under a possible permutation of the labels $\sigma$, it holds that 
\[
\vol\left(\bigcup_{i=1}^k \left\{v | v\in S_i\mbox{\ and\ } \ell_v\neq \sigma(i)  \right\}\right)=o(\vol(V));
\]
\item The total information exchanged among these $n$ nodes, i.e., the message complexity, is $O\left(T\cdot m\cdot \frac{1}{\beta}\log{\frac{1}{\beta}}\right)$ words.
\end{enumerate}
\end{theorem}

As a direct application of the two theorems above, we look at the graph $G$ that consists of $k=O(1)$  expander graphs of almost balanced size connected by sparse cuts. By first applying the sparsification algorithm from Theorem~\ref{main11}, we obtain a sparse subgraph $H$ of $G$ that has a very similar cluster-structure to $G$, and this graph $H$ is obtained with total communication cost $O(n\log n)$ words. Then, we apply the distributed clustering algorithm~(\thmref{cluster}) on  $H$,  which has $O(n\log n)$ edges. The distributed clustering algorithm finishes in $O(\log n)$ rounds,  has total communication cost $O(n\cdot\mathrm{poly}\log n)$ words, and the volume of the misclassified nodes is $o(\vol(V))$. Notice that the communication cost of the two algorithms together is $O(n\cdot\poly\log n)$ words, which is sublinear in the size of $G$ for a dense input graph.

\subsection{Related Work}

There is a large amount of literature on graph clustering, and our work is most closely related to efficient algorithms for graph clustering under different formulations of clusters. Oveis Gharan and Trevisan~\cite{conf/soda/GharanT14} formulate the notion of  clusters with respect to the \emph{inner} and $\emph{outer}$ conductance: a cluster $S$ should have  low outer conductance, and the conductance of the induced subgraph by $S$ should be high.  Under a gap assumption between $\lambda_{k+1}$ and $\lambda_k$, they
  present a polynomial-time algorithm which finds a $k$-way partition $\aseq$ that satisfies the inner and outer conductance condition. 
    

Allen Zhu et al.~\cite{conf/icml/ZhuLM13} studies graph clustering with a gap assumption similar to ours, and presents a local algorithm with better approximation guarantee under the gap assumption.
However, the setup of our algorithms differs significantly from most local graph clustering algorithms~\cite{GharanT12,conf/icml/ZhuLM13,SpielmanT13}  for the following reasons:  (1) We need to run a local algorithm $k$ times in order to find $k$ clusters. However, as the output of each execution of a local algorithm only returns an \emph{approximate} cluster, 
 the approximation ratio of the final output cluster might not be guaranteed when the value of $k$ is large. (2) For many instances, our algorithm requires only a poly-logarithmic number of rounds, while  local algorithms run in time proportional to the volume of the output set. It is unclear how these algorithms could finish in a poly-logarithmic number of rounds, even if we were able to implement them in the distributed setting. 
 

 Becchetti et al.~\cite{becchetti15} studies a distributed process to partition an almost-regular graph into clusters, and their analysis focuses mostly on graphs generated randomly from stochastic block models. 
In contrast to ours, their algorithm requires every node to exchange information with all of its neighbours in each round, and thus has significantly higher communication cost. Moreover,  the design and analysis of our algorithm succeeds to overcome  their regularity constraint as well by an alternative averaging rule.

We  notice that the distributed algorithm presented in Kempe et al.~\cite{kempe04} for computing the top $k$ eigenvectors of the adjacency matrix of a graph can be applied for graph clustering.
Their algorithm, however,  is much more involved than ours. 
Moreover, for an input graph $G$ of $n$ nodes, the number of rounds required in their algorithm is proportional to the mixing time of a random walk in $G$. 
For a graph consisting of multiple expanders connected by very few edges, their algorithm requires $O(\mathrm{poly}(n))$ rounds, 
which is much higher than $O(\poly\log n)$ rounds needed for our algorithm.

Another line of research closely related to our work is graph sparsification, including both   cut sparsification~\cite{benczur_karger} and 
spectral sparsification~\cite{BSS,LS15,LS17,SS11,spielmanTengSS}. The constructions of both cut and spectral sparsifiers, however, are quite complicated or require solving Laplacian systems, while our algorithm is simply based on sampling and easy to implement.
The idea of using sparsification to reduce the communication complexity for clustering a graph in the distributed setting is first proposed by \cite{CSWZ16}. They assume the graph is distributed across $s$ servers, while our work considers more extreme distributed settings: each node of the graph is a computational unit. Our algorithms, however, work in their distributed model as well. Furthermore, we emphasise that the sparsification schemes of \cite{CSWZ16} require the computation of effective resistances, which is very expensive in practice, while our scheme is much simpler and faster.

\subsection{Organisation}

The remaining part of the paper is organised as follows: \secref{pre} lists the  notations used in the paper. We present and analyse the sparsification algorithm in \secref{sparse}, and prove Theorem~\ref{main11}. \secref{cluster} is to present the distributed algorithm for graph clustering, which corresponds to \thmref{cluster}.  We report the experimental results of our sparsification algorithm in \secref{exp}.

\section{Preliminaries}\label{sec:pre}

Let $G=(V,E,w)$ be an undirected weighted graph with $n$ nodes and weight function $w:E \rightarrow \R_{\ge 0}$. For any node $u$, the degree $d_u$ of $u$ is defined as 
$d_u\triangleq\sum_{u\sim v}w(u,v)$, where we write $u\sim v$ if $\{u,v\}\in E[G]$.
For any set $S \subseteq V$, 
 the volume of $S$ is defined by $\vol_G(S) \triangleq \sum_{v\in S} d_v$.
The (normalised) indicator vector of a set $S\subset V$ is defined by $\chi_S\in\mathbb{R}^n$, where   $\chi_S(v)=\sqrt{d_v/\vol(S)}$ if $v\in S$, and $\chi_S(v)=0$ otherwise.

 We work with algebraic objects related to $G$. 
 Let $\mat{A}_G$ be the adjacency matrix of $G$ defined by $(\mat{A}_G)_{u,v}=w(u,v)$ if $\{u,v\}\in E(G)$, and $(\mat{A}_G)_{u,v}=0$ otherwise. The degree matrix 
 $\mat{D}_G$ of $G$  is a diagonal matrix defined by $(\mat{D}_G)_{u,u} = d_u$, and the normalised Laplacian of $G$ is defined by $\calL_G \triangleq \Id - \mat{D}_G^{-1/2} \mat{A}_G \mat{D}_G^{-1/2}$. Alternatively, we can write the normalised Laplacian with respect to the indicator vectors of nodes: for each node $v$,  we define an indicator vector $\chi_v\in\mathbb{R}^n$ by $\chi_v(u)=1/\sqrt{d_v}$  if $u=v$, and $\chi_v(u)=0$ otherwise. We further define $b_e\triangleq \chi_u-\chi_v$ for each edge $e=\{u,v\}$, where the orientation of $e$ is chosen arbitrarily. Then, we can write 
   $\calL_G = \sum_{e=\{u,v\} \in E} w(u,v)\cdot b_e b_e^{\intercal}$. 
   We always use 
$0 = \lambda_1\leq\cdots\leq\lambda_n \le 2$
to express the eigenvalues of $\calL_G$, with 
 their corresponding orthonormal eigenvectors  $f_1,\ldots, f_n$. 
With a slight abuse of notation, we use $\calL_G^{-1}$ for the pseudoinverse of $\calL$, i.e.,  $\calL_G^{-1} \triangleq \sum_{i=2}^n \frac{1}{\lambda_i} f_i f_i^{\intercal}$. 
When $G$ is connected, it holds that 
  $\lambda_2 > 0$ and the matrix $\calL_G^{-1}$
 is well-defined.   Sometimes we drop the subscript $G$ when it is clear from the context.
  
Remember that 
the Euclidean norm of any vector $x\in\R^n$ is defined as $\|x\|\triangleq\sqrt{\sum_{i=1}^n x_i^2}$, and the spectral norm of any matrix $\mat{M} \in \R^{n \times n}$ is defined as 
\[
\|\mat{M}\| \triangleq \max_{x \in \R^n\setminus\{\mathbf{0}\}} \frac{\|\mat{M}x\|}{\|x\|}.
\] 
\section{Cluster-Preserving Sparsifiers}
\label{sec:sparse}

 In this section we present an algorithm for constructing a cluster-preserving sparsifier that can be easily implemented in the distributed setting.  Our algorithm is based on sampling edges with respect to the degrees of their endpoints, which was originally introduced in \cite{spielmanTengSS} as a way to construct spectral sparsifiers for graphs with \emph{high} spectral expansion. To sketch the intuition behind our algorithm, let us look at the following toy example illustrated in Figure~\ref{ssfig}, i.e., the graph $G$ consisting of two complete graphs of $n$ nodes connected by a single edge.
It is easy to see that, when we sample $O(n\log n)$ edges uniformly at random from $G$ to form a graph $H$,  with high probability the middle edge will not be sampled and $H$ will consists of 
 two isolated expander graphs, each of which has  constant spectral expansion. Although our sampled graph  $H$ does not preserve the  spectral and cut structure of $G$, it does preserve its cluster-structure: every reasonable clustering algorithm will recover these two disjoint components of $H$,  which correspond exactly to the two clusters in $G$. 
 We will show that this sampling scheme can be generalised, and 
sampling every edge $u\sim v$ with  probability depending only on $d_u$ and $d_v$ suffices to construct a sparse subgraph that preserves the cluster-structure of the original graph.


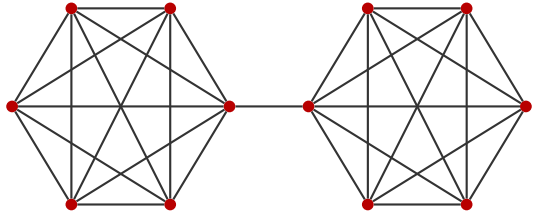
\begin{figure}[h]
\begin{center}
\begin{tikzpicture}[xscale=1.3,yscale=1.3,rounded corners=3pt,knoten/.style={fill=ocre,color=ocre,circle,scale=0.23},edge/.style={black!80, thick},tedge/.style={black!80, line width=2pt}]

\node[knoten] (1) at (0,0) {$1$};
\node[knoten] (2) at (1,0) {$2$};

\node[knoten] (3) at (-0.6,1) {$3$};
\node[knoten] (4) at (1.6,1) {$4$};

\node[knoten] (5) at (0,2) {$3$};
\node[knoten] (6) at (1,2) {$4$};

\draw[edge] (1) -- (2);
\draw[edge] (1) -- (3);
\draw[edge] (1) -- (4);
\draw[edge] (1) -- (5);
\draw[edge] (1) -- (6);

\draw[edge] (2) -- (3);
\draw[edge] (2) -- (4);
\draw[edge] (2) -- (5);
\draw[edge] (2) -- (6);

\draw[edge] (3) -- (4);
\draw[edge] (3) -- (5);
\draw[edge] (3) -- (6);

\draw[edge] (4) -- (5);
\draw[edge] (4) -- (6);

\draw[edge] (5) -- (6);

\node[knoten] (11) at (3,0) {$1$};
\node[knoten] (12) at (4,0) {$2$};

\node[knoten] (13) at (2.4,1) {$3$};
\node[knoten] (14) at (4.6,1) {$4$};

\node[knoten] (15) at (3,2) {$3$};
\node[knoten] (16) at (4,2) {$4$};

\draw[edge] (13) -- (4);

\draw[edge] (11) -- (12);
\draw[edge] (11) -- (13);
\draw[edge] (11) -- (14);
\draw[edge] (11) -- (15);
\draw[edge] (11) -- (16);

\draw[edge] (12) -- (13);
\draw[edge] (12) -- (14);
\draw[edge] (12) -- (15);
\draw[edge] (12) -- (16);

\draw[edge] (13) -- (14);
\draw[edge] (13) -- (15);
\draw[edge] (13) -- (16);

\draw[edge] (14) -- (15);
\draw[edge] (14) -- (16);

\draw[edge] (15) -- (16);

\end{tikzpicture}
\end{center}

\caption{The graph $G$ consists of two complete subgraphs of $n$ nodes connected by an edge. It is
easy to see that sampling $O(n \log n)$ edges uniformly at random suffices to construct a subgraph having the same cluster-structure of $G$.
\label{ssfig}}
\end{figure}


 \subsection{Algorithm Description}

In our algorithm  every node $u$ checks every edge $e=\{u,v\}$ adjacent to $u$ itself, and samples edge $e$ with probability 
\begin{equation}\label{eq:samprob}
p_u(v)\triangleq \min\left\{w(u,v)\cdot  \frac{\tau \; \log n}{d_u} ,1\right\} 
\end{equation}
for some parameter $\tau$ satisfying $\tau \ge C / \lambda_{k+1}$ for a large enough constant $C\in\mathbb{R}_{\geq 0}$.
The algorithm uses a set $F$ to maintain all the sampled edges, where $F$ is initially set to be empty. Finally, the algorithm returns a weighted graph $H=(V,F, w_H)$, where the weight $w_H(u,v)$ of every edge $e=\{u,v\} \in F$ is defined as
\[
w_H(u,v)\triangleq \frac{w(u,v)}{p_e},
\]
and
\[
p_e\triangleq p_u(v)+ p_v(u)-p_u(v)\cdot p_v(u)
\]
 is  the probability that $e$ is sampled by at least one of its endpoints. 
 Notice that our algorithm can be easily implemented in a distributed setting:  any node $u$ chooses to retain (or not) an edge $u\sim v$ independently from any other node, and communication between $u$ and $v$ is needed only if  $u\sim v$ is sampled by one of its two endpoints. Therefore, the total communication cost of the algorithm is proportional to the number of edges in $H$.   
 
 

\subsection{Analysis of the Algorithm} 

Now we analyse the algorithm, and prove Theorem~\ref{main11}. At a high level, our proof  consists of the following two steps:
\begin{enumerate}
\item  We analyse the intra-connectivity of the clusters in the returned graph $H$: we show that the top $n-k$ eigenspaces of $\mathcal{L}_G$ are preserved in $\calL_H$, and hence $ \lambda_{k+1}(\mathcal{L}_H) =\Theta(\lambda_{k+1}(\mathcal{L}_G))$.  
\item 
We show that the conductance of $S_1,\ldots, S_k$ are low in $H$, i.e., 
  \begin{equation}
\label{eq:lowconductance}
\phi_H(S_i) = O\left(k\cdot \phi_G(S_i)\right) \text{ for any } i=1,\dots,k.
\end{equation}
\end{enumerate}
Combining these two steps, we will prove that $\Upsilon_H(k)=\Omega(\Upsilon_G(k)/k)$, which proves the approximation guarantees of Theorem~\ref{main11}. 
The bound on the number of edges in   $H$ follows from  the definition of the sampling scheme of our algorithm. 

The following  concentration inequalities  will be used in our proof.

\begin{lemma}[Problem 1.9, \cite{DP09}]\label{lem:chernoff1}
Let $X_1,\ldots, X_n$ be independent random variables such that for each $i\in [n]$, $x_i\in[a_i,b_i]$ for some reals $a_i$ and $b_i$. Then it holds that 
\[
\Pro{|X- \mathbf{E}[X]|\geq t} \leq 2\mathrm{exp} \left(-\frac{2t^2}{\sum_{i} (b_i - a_i)^2} \right).
\]
\end{lemma}

\begin{lemma}[Matrix Chernoff Bound, \cite{tropp}]
\label{lem:chernoff}
Consider a finite sequence $\{X_i\}$ of independent, random, PSD matrices of dimension $d$ that satisfy $\|X_i\| \le R$. Let $\mu_{\min} \triangleq \lambda_{\min}\left(\Ex{\sum_i X_i}\right)$ and $\mu_{\max} \triangleq \lambda_{\max}\left(\Ex{\sum_i X_i}\right)$. Then it holds that 
\begin{align*}
\Pro{\lambda_{\min}\left({\sum_i X_i}\right) \le (1-\delta)\mu_{\min}} &\le d \cdot \left(\frac{\mathrm{e}^{-\delta}}{(1-\delta)^{1-\delta}}\right)^{\mu_{\min}/R} \text{  for } \delta \in [0,1], \text{ and} \\
\Pro{\lambda_{\max}\left({\sum_i X_i}\right) \ge (1 +\delta)\mu_{\max}} &\le d \cdot \left(\frac{\mathrm{e}^{\delta}}{(1+\delta)^{1+\delta}}\right)^{\mu_{\max}/R} \text{  for } \delta \ge 0.
\end{align*}
\end{lemma}

\begin{proof}[Proof of Theorem~\ref{main11}] Without loss of generality we assume within the proof   that it holds for  any edge $u\sim v$ that
\[
w(u,v) \cdot \frac{\tau \cdot \log{n}}{ d_u} < 1.\]
Otherwise, edge $u\sim v$ is always added into $H$ by Algorithm~1, and this will not affect our analysis, as there are $O(n\tau\log n)$ such edges.

 Let $\overline{\calL}_G$ be the projection of $\calL_G$ on its top $n-k$ eigenspaces, i.e., \[
\overline{\calL}_G = \sum_{i=k+1}^n \lambda_i f_i f_i^{\intercal}.
\]
 With a slight abuse of notation we call $\overline{\calL}_{G}^{-1/2}$ the square root of the pseudoinverse of $\overline{\calL}_{G}$, i.e., 
 $
 \overline{\calL}_{G}^{-1/2} =  \sum_{i=k+1}^n (\lambda_i)^{-1/2} f_i f_i^{\intercal}$.
Analogously, we call $\overline{\mathcal{I}}$ the projection on $\Span\{f_{k+1},\dots,f_n\}$, i.e.,
 \[
 \overline{\mathcal{I}}\triangleq \sum_{i=k+1}^n f_i f_i^{\intercal}.
 \]

We will first prove that the top $n-k$ eigenspaces of $\calL_G$ are preserved, which implies that $\lambda_{k+1}(\mathcal{L}_G)=\Theta(\lambda_{k+1}(\mathcal{L}_H))$. 
To prove this statement,  we examine the properties of the graph $H$ constructed by the algorithm.  Remember that for any $e=\{u,v\}$ we have that 
\[
p_e=p_u(v) + p_v(u)- p_u(v)\cdot p_v(u),
\]
and it holds that 
 $\frac{1}{2}(p_u(v)+p_v(u)) \le p_e \le p_u(v) + p_v(u)$. Now fo each edge $e=\{ u,v\}$ of $G$ we define a random matrix $X_e\in\mathbb{R}^{n\times n}$ by 
 \[
X_e = 
\begin{cases}
  w_H(u,v)\cdot  \overline{\calL}_{G}^{-1/2} b_{e} b_{e}^{\intercal} \overline{\calL}_{G}^{-1/2} & \text{if\ } u\sim v \text{\ is sampled by the algorithm}, \\
  \mat{0}          & \text{otherwise.}\ 
\end{cases}
\]
Notice that
\[
\sum_{e\in E[G]} X_e = \sum_{ \mathrm{sampled\ edges\ } e=\{u,v\}} w_H(u,v)\cdot  \overline{\calL}_{G}^{-1/2} b_{e} b_{e}^{\intercal} \overline{\calL}_{G}^{-1/2}= \overline{\calL}_{G}^{-1/2} \calL_H \overline{\calL}_{G}^{-1/2},
\]
and 
\[
\Ex{\sum_{e \in E} X_e} = \sum_{e=\{u,v\} \in E[G]} p(e)\cdot w_H(u,v)
\cdot  \overline{\calL}_{G}^{-1/2} b_{e} b_{e}^{\intercal} \overline{\calL}_{G}^{-1/2} = \overline{\calL}_G^{-1/2} \calL_G \overline{\calL}_G^{-1/2} = \overline{\mathcal{I}}.\] 
Moreover, for any sampled $e =\{u,v\} \in E$ we have that 
\begin{align*}
\|X_e\| 	&\leq   w_H(u,v)\cdot b_{e}^{\intercal} \overline{\calL}_{G}^{-1/2} \overline{\calL}_{G}^{-1/2} b_{e} =  
		\frac{w(u,v)}{p_{e}}\cdot b_{e}^{\intercal} \overline{\calL}_{G}^{-1}  b_{e} \le \frac{w(u,v)}{p_e} \left\|\overline{\calL}_{G}^{-1}\right\| \|b_e\|^2  \\
		&\le \frac{2}{\tau  \log{n}\cdot\left(\frac{1}{d_u}+\frac{1}{d_v}\right)} \cdot \frac{1}{\lambda_{k+1}}  \left(\frac{1}{d_u}+\frac{1}{d_v}\right) \le \frac{2}{C \log{n}},
\end{align*}
where the last inequality follows by the fact that $\tau\geq C/\lambda_{k+1}$.
To apply the matrix Chernoff bound, since we work on the top $(n-k)$ eigenspaces, we can assume for simplicity that $\mu_{\min}=1$. Therefore, by setting $R =  \frac{2}{C \log{n}}$, $\delta = 1/2$,  $\mu_{\min} = \mu_{\max} = 1$, the matrix Chernoff bound~(\lemref{chernoff}) gives us that 
\[
\mathbf{P}\left[\lambda_{\min}\left(\sum_{e\in E[G]} X_e \right) \geq 1/2 \right] =1-O(1/n^c),
\]
and 
\[
\mathbf{P}\left[\lambda_{\max}\left(\sum_{e\in E[G]} X_e \right) \leq 3/2 \right] =1-O(1/n^c),
\]
for some constant $c$.  Combining this with the Courant-Fischer theorem and $$\Dim(\Span\{f_{k+1},\dots,f_n\}) = n-k,$$ we have that 
$
\lambda_{k+1}(\mathcal{L}) =\Theta(\lambda_{k+1}(\mathcal{L}_G))$.


Now we analyse the conductance of every cluster $S_i$ in $H$. For any edge $e=\{u,v\}$ we define a random variable such that \[
Y_e = 
\begin{cases}
  \frac{w(u,v)}{p_{e}}  & \text{with probability}\ p_{e}, \\
  0          & \text{otherwise.}\ 
\end{cases}
\]
Hence, it holds for any $1\leq i\leq k$ that 
\[
\Ex{w_H(S_i, V\setminus S_i)}=\Ex{\sum_{\substack{e=\{u,v\}, \\ u\in S_i, v\not\in S_i}} Y_{\{u,v\}}}=w(S_i, V\setminus S_i).
\]
 Hence, by Markov's inequality and the union bound, with constant probability it holds for all $i=1,\ldots, k$ that  
 \begin{equation}\label{eq:upcond}
 w_H(S_i, V\setminus S_i) = O\left(k \cdot  \delta_G(S_i) \right).
 \end{equation}

We further analyse $\vol_H(S_i)$. For any  $u\in S_i$, let 
\[
d^{\star}_G(u) = \sum_{\substack{ u\sim v, v\in S_i\\  d_u\leq d_v}} w(u,v),
\]
and 
\[
\vol^{\star}_G(S_i) = \sum_{u\in S_i} d^{\star}_G(u).
\]
Since $w(u,v)$ for any  internal edge $u\sim v$ in $S_i$  contributes only twice to $\vol_G(S_i)$ and $\phi_G(S_i) \le 1/2 $, it 
always
holds that $\vol^{\star}_G(S_i) \geq \vol_G(S_i)/4$.   
We define random variables 
\[
d^{\star}_H(u) =  \sum_{\substack{ u\sim v, v\in S_i\\  d_u\leq d_v}} Y_e,
\]
and
\[
\vol^{\star}_H(S_i) = \sum_{u\in S_i} d^{\star}_H(u).
\]
By definition, we have that $\Ex{\vol^{\star}_H(S_i)} = \vol^{\star}_G(S_i)$.
Since it holds for any $e=\{u,v\}$ that
\[
0 \le \frac{w(u,v)}{p_e} \leq \frac{2 w(u,v)}{p_u+ p_v} = \frac{2}{\tau\log n \cdot \left( 1/d_u + 1/d_v\right)},
\]
 we can  apply \lemref{chernoff1} and obtain that
\begin{align*}
\lefteqn{\Pro{ \left|\vol^{\star}_H(S_i) - \vol^{\star}_G(S_i)\right| \geq 1/2\cdot \vol^{\star}_G(S_i)}}\\
& = 2\cdot\mathrm{exp} \left( - \frac{1/2\cdot (\vol_G^{\star}(S_i))^2}{ \sum_{u\in S_i} \sum_{u\sim v, v\in S_i, d_u\leq d_v} \left(\frac{2}{\tau\log n \cdot \left( 1/d_u + 1/d_v\right)}\right)^2} \right) \\
& \leq 2\cdot\mathrm{exp} \left( - \frac{\tau^2\log^2n\cdot (\sum_{u\in S_i} d_u)^2/100 }{ \sum_{u\in S_i} \sum_{u\sim v, v\in S_i, d_u\leq d_v } \left(\frac{d_ud_v}{ d_u + d_v}\right)^2} \right)\\
& \leq  2\cdot\mathrm{exp} \left( - \frac{\tau^2\log^2n\cdot (\sum_{u\in S_i} d_u)^2/100 }{ \sum_{u\in S_i} \sum_{u\sim v, v\in S_i, d_u\leq d_v } d_ud_v} \right)\\
& = O(1/n^c)
\end{align*}
for some constant $c\geq 10$. Hence, with probability $1-O(1/n^c)$ it holds that \
\begin{equation}\label{eq:lbvol}
\vol_H^{\star}(S_i)\geq \frac{1}{2}\cdot \vol_G^{\star}(S_i)\geq \frac{1}{8}\cdot\vol_G(S_i).
\end{equation}
By the union bound, \eq{lbvol} holds for all the $k$ clusters. Combining this with \eq{upcond} shows that $\phi_H(S_i) = O(k\cdot \phi_G(S_i))$ for any $1\leq i\leq k$, and 
$\Upsilon_H(k)=\Omega(\Upsilon_G(k)/k)$. 
 \end{proof}

\section{Distributed Graph Clustering}
\label{sec:cluster}

In this section we present and analyse a distributed algorithm to partition a graph $G$ that possesses a cluster-structure with clusters of balanced size, and prove \thmref{cluster}.

 \subsection{Algorithm}

Our algorithm consists of  \textsc{Seeding}, \textsc{Averaging}, and \textsc{Query} steps, which are described as follows.
  
\textbf{The \textsc{Seeding} step:}  The algorithm sets \[
\bar{s}=\Theta\left(\frac{1}{\beta}\cdot \log \frac{1}{\beta}\right),
\]
 and each node $v$ chooses to be \emph{active} with probability $\bar{s}\cdot d_v/\vol(V)$.  For simplicity, we assume that $v_1,\cdots, v_s$ are the active nodes, for some $s\in\mathbb{N}$. 
The algorithm associates each active node with a vector $x^{(0,i)}=\chi_{v_i}$, and these 
 vectors $x^{(0,1)},\ldots, x^{(0,s)}$ represent the initial state~(round $0$) of the graph, where each node $v$ only maintains the values $x^{(0,1)}(v),\ldots, x^{(0,s)}(v)$. Notice that the information about which nodes are  active  doesn't need to be broadcasted during the seeding step of the algorithm.

\textbf{The \textsc{Averaging} step:} This step consists of $T$ rounds, and in each round every node $v$ updates its state based on the states of its neighbours from the previous round. Namely, for any $1\leq i\leq s$, the values $x^{(t,i)}(v)$ maintained by node $v$ in round $t$ are computed according to
\begin{equation}
\label{eq:update}
x^{(t,i)}(v) = \frac{1}{2} x^{(t-1,i)}(v) +  \frac{1}{2} \sum_{\{u,v\} \in E} \frac{w(u,v)}{\sqrt{d_u d_v}} x^{(t-1,i)}(u).
\end{equation}


\textbf{The \textsc{Query} step:} Every node $v$ computes the label $\ell_v$ of the cluster that it belongs to by the formula 
\begin{equation}\label{eq:query}
\ell_v = \min\left\{i \, \mid \, x^{(t,i)}(v) \ge  {\frac{\sqrt{d_v}}{{2} \beta \vol(V)}}\right\}.
\end{equation}


Notice that the execution of the algorithm requires each node to know certain parameters about the graph, including the number of nodes $n$, the  volume of the graph $\vol(V)$, a bound $\beta$ on the size of the clusters, and the value of $T$. 
However,  nodes  do not need to know the exact values of these parameters, but only a reasonable approximation. Moreover,  although the value of  $T$ is application-dependent, for graphs with  clusters that have  strong intra-connectivity properties,  we can set $T \approx \log{n}$ in practice.  
 
\subsection{Analysis of the Algorithm}



In this section we analyse the distributed clustering algorithm, and prove Theorem~\ref{thm:cluster}. Remember that we assume that $G$ has an optimal clustering $S_1,\ldots, S_k$ with $\vol(S_i) \ge \beta \vol(V)$ for any $1\le i \le k$, and $G$ satisfies the following gap assumption:
\begin{equation}
\label{eq:gap}
\Upsilon_G(k)= \omega\left( k^4 \log^2\frac{1}{\beta} + \log{n} \right).
\end{equation}

Before analysing the algorithm, we first discuss some intuitions behind the proof. Remember that 
the configuration of the network in  round $t$ of the averaging step is expressed by $s$ vectors $x^{(t,1)},\ldots, x^{(t,s)}$, and these vectors are updated according to \eq{update}.  For the sake of intuition, we assume that $G$ is regular, and then 
 the vector $x^{(t,i)}$ corresponds to the probability distribution of a  $t$-step lazy random walk in $G$. It is well-known that the vector
 $x^{(t,i)}$ converges to the uniform distribution as $t$ tends to infinity. The time $T=\Theta(\log{n}/\lambda_{k+1})$, instead, corresponds to the \emph{local mixing time} of the clusters: if a random walk starts with $v_i \in S_i$, then the probability distribution of this $T$-step random walk
 will be mixed (uniform) inside $S_i$, conditioned on the fact that the random walk never leaves that cluster.  Our analysis shows that, when picking $v_i$ at random from $S_i$, with high probability the distribution of the random walk after $T$ steps is concentrated on $S_i$.  In other words, after $T$ rounds, each vector $x^{(T,1)},\ldots, x^{(T,s)}$ is almost uniform on one of the clusters, and close to zero everywhere else. Hence, as long as we  hit all the clusters with at least one initial active node, the query step will assign the same label to two nodes if and only if they belong to the same cluster (for most pairs of nodes).

When $G$ is not regular, \eq{update} suggests that the averaging step can be thought as a \emph{power iteration method} to approximate ($k$ linearly independent combination of) the bottom eigenvectors of $\calL_G$. We will show that these eigenvectors contain all the information needed to obtain a good partitioning of the graph.

To formalise the intuitions before, similar with the definition of $\overline{\mathcal{I}}$, let
\[
\underline{\mathcal{I}} = \sum_{i=1}^k f_i f_i^{\intercal}
\]
 be the projection on the bottom $k$ eigenspaces. We first prove that, starting the process with a single initial vector $x^{(0)}$,  $x^{(T)}$ is close to $\underline{\mathcal{I}} x^{(0)}$. 

\begin{lemma}
\label{lem:approx}
For a large constant $c>0$, it holds
 \[
 \left\| x^{(T)}- \underline{\mathcal{I}} x^{(0)} \right\| = O\left(\frac{\log{n} }{\Upsilon_G(k)}\cdot \left\|\underline{\mathcal{I}}  x^{(0)}\right\| + n^{-c}\right).
 \]
\end{lemma}
\begin{proof}

By the update rule of the algorithm, we have that $x^{(T)} = \mat{P}^T x^{(0)}$, where 
\[
\mat{P} = \frac{1}{2}\cdot \mat{I} + \frac{1}{2}\cdot\mat{D}^{-1/2} \mat{A}\mat{D}^{-1/2} = \mat{I} - \frac{1}{2}\mathcal{L}_G.
\]
Hence, it holds that 
\[
\mat{P}^T = \left(\Id - \frac{1}{2}\calL_G\right)^T = \sum_{i=1}^n \left(1 - \frac{\lambda_i}{2}\right)^T f_i f_i^{\intercal}.
\]
Notice that for $i > k$ it holds that 
\[
\left(1 - \frac{\lambda_i}{2}\right)^T \le \left(1 -\frac{ \lambda_{k+1}}{2}\right)^{2c \cdot \log{n}/\lambda_{k+1}} \le \mathrm{e}^{-c \log{n}} \le n^{-c},
\]
while for $i=1,\dots,k$ it holds that 
\[
\left(1 - \frac{\lambda_i}{2}\right)^T \ge 1 - \frac{T\cdot \lambda_i}{2}
 \ge 1 - O\left(\frac{\log{n} \cdot  \lambda_i}{\lambda_{k+1}}\right) \ge 1 - O\left(\frac{\log{n} \cdot  \lambda_k}{\lambda_{k+1}}\right).
\]
Hence,
\begin{align*}
\left\| x^{(T)} - \underline{\mathcal{I}} x^{(0)} \right\|^2 &= \left\|\left( \mat{P}^T - \underline{\mathcal{I}}\right) x^{(0)}\right\|^2 \\
						&= \sum_{i=1}^k \left(1 - \left(1 - \frac{\lambda_i}{2}\right)^{T}\right)^2   \left\langle f_i, x^{(0)} \right\rangle^2  
								+  \sum_{i=k+1}^n \left(1 - \frac{\lambda_i}{2}\right)^{2T}  \left\langle f_i, x^{(0)} \right\rangle^2\\
			&=	O\left(\left(\frac{\log{n} \cdot  \lambda_k}{\lambda_{k+1}}\right)^2 \sum_{i=1}^k  \left\langle f_i, x^{(0)} \right\rangle^2 + n^{-2c} \right)\\
				&= O\left(\left(\frac{\log{n} \cdot  \lambda_k}{\lambda_{k+1}}\left\|\underline{\mathcal{I}}x^{(0)}\right\|\right)^2 + n^{-2c}\right)\\
				&= O\left(\left(\frac{\log{n} }{\Upsilon_G(k)}\left\|\underline{\mathcal{I}}x^{(0)}\right\|\right)^2 + n^{-2c}\right),
\end{align*}
where the last inequality follows from the higher-order Cheeger inequality. Then, taking the square root on both sides of the equality above proves the lemma.
\end{proof}


As the goal is to use $x^{(T)}$ to recover the clusters, we need to relate $\underline{\mathcal{I}} $ to their indicator vectors. The following result proves that the bottom $k$ eigenvectors of $\mathcal{L}_G$ are close to a linear combination of the indicator vectors of  $S_1,\ldots, S_k$.

\begin{lemma}
\label{lem:structure}
Let  $\Upsilon_G(k)=\Omega\left(k^2\right)$. For any $1\leq i\leq k$ there exists $\widehat{\chi}_i \in \Span\{\chi_{S_1},\dots,\chi_{S_k}\}$ such that
$
\left\|\widehat{\chi}_i - f_i\right\| = O\left(k \sqrt{\frac{k}{\Upsilon_G(k)}}\right)$.
Moreover $\{\widehat{\chi}_i\}_{i=1}^k$ form an orthonormal set.
\end{lemma}
To prove \lemref{structure}, we need the following lemma:

\begin{lemma}[\cite{peng15}]
\label{lem:structure1}
Let $\{S_i\}_{i=1}^k$ be a $k$-way partition of $G$ achieving $\rho(k)$, and let $\Upsilon_G(k)=\Omega\left(k^2\right)$. Assume that 
 $\tilde{\chi}_i$ is  the projection of $f_i$ in the span of $\{\chi_{S_1},\dots,\chi_{S_k}\}$. Then, it holds  for any $1\leq i\leq k$ that
\[
\left\|\tilde{\chi}_i - f_i\right\| = O\left(\sqrt{\frac{k}{\Upsilon_G(k)}}\right).
\]
\end{lemma}

\begin{proof}[Proof of \lemref{structure}]
Since $\{{f}_i\}_{i=1}^k$ is an orthonormal set, it holds by \lemref{structure1} that $\{\tilde{\chi}_i\}_{i=1}^k$ are \emph{almost} orthonormal. 
Hence, our task is to construct an orthonormal set $\{\widehat{\chi}_i\}_{i=1}^k$ based on $\{\tilde{\chi}_i\}_{i=1}^k$, which can be achieved by applying the Gram-Schmidt orthonormalisation procedure. The error bound follows from the fact that 
\begin{align*}
\left\langle \tilde{\chi}_i, \tilde{\chi}_j \right\rangle & = \frac{1}{2}\cdot \left( \|\tilde{\chi}_i \|^2 + \| \tilde{\chi}_j\|^2 - \|\tilde{\chi}_i - \tilde{\chi}_j \|^2  \right)= O\left(\sqrt{\frac{k}{\Upsilon_G(k)}}\right)
\end{align*}
holds for $i\neq j$. 
\end{proof}

Based on  \lemref{structure},  we will prove in the next lemma that,  for any cluster $S_1,\dots,S_k$ and  for most starting nodes $v \in S_j$,  $x^{(T)} $ is close to $\chi_{S_j}$.
\begin{lemma}
\label{lem:good}
Let $A \subseteq V$ be the subset of nodes such that, for any $j=1,\dots,k$ and any $v \in A \cap S_j$, setting $x^{(0)} = \chi_v$ we have that
\[
\left\| x^{(T)} - \frac{1}{\sqrt{\vol(S_j)}} {\chi}_{S_j}\right\| = O\left(\sqrt{\frac{ k^4 \log{(1/\beta)}}{\Upsilon_G(k) \beta \vol(V)}}\right).
\]
Then, it holds that 
\[
\vol(A) \ge \vol(V)\left(1 - \frac{\beta}{C\log (1/\beta)}\right),
\]
for some constant $C$.
\end{lemma}
\begin{proof}
Without loss of generality we assume
 $v \in S_j$, and let $\{\widehat{\chi}_i\}_{i=1}^k$ be the set of vectors  defined  in \lemref{structure}. We show that the projection of $\chi_v$ on $\Span\{\widehat{\chi}_1,\dots,\widehat{\chi}_k\}$ is exactly equal to $\frac{1}{\sqrt{\vol(S_j)}}\chi_{S_j}$. To this end, first notice that $\Span\{\widehat{\chi}_1,\dots,\widehat{\chi}_k\} = \Span\{{\chi}_{S_1},\dots,{\chi}_{S_k}\}$, since each $\widehat{\chi}_i$ is by definition a linear combination of vectors in $\{{\chi}_{S_i}\}_{i=1}^k$ and 
 \[
 \Dim\left(\Span\{\widehat{\chi}_1,\dots,\widehat{\chi}_k\}\right) = \Dim\left(\Span\{{\chi}_{S_1},\dots,{\chi}_{S_k}\}\right) = k.
 \]
  Then,
\begin{equation}
\label{eq:chistar}
  \sum_{i=1}^k \left\langle \chi_v, \widehat{\chi}_i \right\rangle \widehat{\chi}_i =  \sum_{i=1}^k \left\langle \chi_v, {{\chi}_{S_i}} \right\rangle {{\chi}_{S_i}} = \left\langle \chi_v, {\chi}_{S_j} \right\rangle {{\chi}_{S_j}} = \frac{1}{\sqrt{\vol(S_j)}} {\chi}_{S_j}.
\end{equation}
where the first equality holds by the fact that $\Span{\{\widehat{\chi}_1,\dots,\widehat{\chi}_k\}} = \Span{\{{\chi}_{S_1},\dots,{\chi}_{S_k}\}}$ and the orthonormality of the two sets of vectors, and the second holds because $\chi_v$ is orthogonal to every $\chi_{S_{\ell}}$ with $\ell \neq j$.

By the triangle inequality we have
\begin{equation}
\label{eq:separate}
\left\| x^{(T)} - \frac{1}{\sqrt{\vol(S_j)}} {\chi}_{S_j}\right\| \le \left\| x^{(T)} -\underline{\mathcal{I}}\chi_v\right\| + \left\| \underline{\mathcal{I}}\chi_v - \frac{1}{\sqrt{\vol(S_j)}} {\chi}_{S_j}\right\|.
\end{equation}
By \lemref{approx} we have 
\begin{equation}
\label{eq:bound0}
 \left\| x^{(T)} - \underline{\mathcal{I}}\chi_v\right\|  = O\left(\frac{\log{n} \cdot  \left\|\underline{\mathcal{I}}\chi_v\right\| }{\Upsilon_G(k)} + n^{-c}\right).
\end{equation}

For any $v\in V$, let 
 \[
 \alpha_v \triangleq \sqrt{\frac{1}{{d_v}}\sum_{i=1}^k \left(f_i(v) - \widehat{\chi}_i(v)\right)^2}.
 \]
 Then it holds that 
\begin{align}
\left\|\underline{\mathcal{I}}\chi_v\right\|^2
	&\;\;=\sum_{i=1}^k \left \langle \chi_v, f_i \right \rangle^2 \nonumber =\sum_{i=1}^k \left \langle \chi_v, \widehat{\chi}_i - \left(\widehat{\chi}_i - f_i\right) \right \rangle^2 \nonumber\\
	&\;\;=\sum_{i=1}^k \left(\left \langle \chi_v, \widehat{\chi}_i\right \rangle - \left \langle \chi_v, \widehat{\chi}_i - f_i \right \rangle \right)^2 \nonumber\\
	&\;\;\le \sum_{i=1}^k 2 \left(\left \langle \chi_v, \widehat{\chi}_i\right \rangle^2 + \left \langle \chi_v, \widehat{\chi}_i - f_i \right \rangle^2 \right) \label{eq:normbound1} \\
	&\;\;=  \frac{2}{\vol(S_i)} + 2 \sum_{i=1}^k \left \langle \chi_v, \widehat{\chi}_i - f_i \right \rangle^2  \label{eq:normbound2}\\
	&\;\;\leq   \frac{2}{\vol(S_i)}  + 2 \alpha_v^2  \label{eq:normbound3} 
\end{align}
where \eq{normbound1} follows from the inequality $(a-b)^2 \le 2(a^2+b^2)$, \eq{normbound2} follows from  \eq{chistar}, and \eq{normbound3} follows from the definition of $\alpha_v$. Hence, it holds that
\begin{align}
\label{eq:normprojbound}
\left\|\underline{\mathcal{I}}\chi_v\right\|= O\left( \frac{1}{\sqrt{\vol(S_i)}}  + \alpha_v \right).
\end{align}

For the second term in the right hand side of \eq{separate}, by \eq{chistar} and the triangle inequality we have that 
\begin{align}
\lefteqn{\left\| \underline{\mathcal{I}}\chi_v - \frac{1}{\sqrt{\vol(S_j)}} {\chi}_{S_j}\right\|} \nonumber\\
&  =  \left\|\sum_{i=1}^k  \langle \chi_v,f_i \rangle f_i - \sum_{i=1}^k  \langle \chi_v,f_i \rangle \widehat{\chi}_i 
								+  \sum_{i=1}^k\langle \chi_v,f_i \rangle \widehat{\chi}_i - \sum_{i=1}^k  \langle \chi_v,\widehat{\chi}_i \rangle \widehat{\chi}_i \right\| \nonumber \\
								& \leq  \left\|\sum_{i=1}^k  \langle \chi_v,f_i \rangle f_i - \sum_{i=1}^k  \langle \chi_v,f_i \rangle \widehat{\chi}_i \right\|
								+  \left\|\sum_{i=1}^k\langle \chi_v,f_i \rangle \widehat{\chi}_i - \sum_{i=1}^k  \langle \chi_v,\widehat{\chi}_i \rangle \widehat{\chi}_i \right\|  \label{eq:sep}
\end{align}
 Let's now analyse the two terms in \eq{sep} separately. For the first term, it holds that
 \begin{align}
 \left\|\sum_{i=1}^k  \langle \chi_v,f_i \rangle f_i - \sum_{i=1}^k  \langle \chi_v,f_i \rangle \widehat{\chi}_i \right\| &\le
  			\sum_{i=1}^k \left| \langle \chi_v, f_i \rangle \right| \left\|f_i - \widehat{\chi}_i \right\|  \nonumber\\
			&\le O\left(\sqrt{\frac{k^3}{\Upsilon_G(k)}}\right)\cdot \sum_{i=1}^k \left| \langle \chi_v, f_i \rangle \right| \le \sqrt{\frac{k^4}{\Upsilon_G(k)}} \left\|\underline{\mathcal{I}}\chi_v \right\|  \nonumber \\
			&= O\left(\sqrt{\frac{k^4}{\Upsilon_G(k) \beta\vol(V)}} + \alpha_v\right) \label{eq:sep00}
 \end{align}
 where the first inequality follows from the triangle inequality, the second by \lemref{structure} and the Cauchy-Schwarz inequality, and the last inequality follows  by \eq{normprojbound}. For the second term of \eq{sep}, we have
\begin{align}
\left\|\sum_{i=1}^k \left( \langle \chi_v,f_i \rangle - \langle \chi_v,\widehat{\chi}_i \rangle \right) \widehat{\chi}_i  \right\| 
 	&= \left\|\sum_{i=1}^k \frac{1}{\sqrt{d_v}}\left( f_i(v) -  \widehat{\chi}_i(v)  \right) \widehat{\chi}_i \right\| = \sqrt{\sum_{i=1}^k\frac{1}{{d_v}} \left(f_i(v) - \widehat{\chi}_i(v)\right)^2} = \alpha_v,\label{eq:sep0}
\end{align}
 where the second equality follows from the orthonormality of $\{\widehat{\chi}_i\}_i$. 
Putting all these inequalities above together, we have 
\begin{align}
\left\| x^{(T)} - \frac{1}{\sqrt{\vol(S_j)}} {\chi}_{S_j}\right\| 
&= O\Biggl(\frac{\log{n}}{\Upsilon_G(k) \cdot \sqrt{\beta \vol(V)}} + n^{-c} + \sqrt{\frac{k^4}{\Upsilon_G(k) \beta\vol(V)}} 
				 + \alpha_v \Biggr) . \label{eq:sep1}
\end{align}
 
 Let's define the set
 \[
 B\triangleq \left\{v| \alpha_v > \sqrt{\frac{C k^4\log{(1/\beta)}}{\Upsilon_G(k) \beta \vol(V)}} \right\},
 \]
 for some constant $C$. Let us look at $\vol(B)$ now.
 By \lemref{structure} we know that 
 \[
 \sum_{v \in V} d_v  \alpha_v^2 =  \sum_{i=1}^k \sum_{v \in V} \left(f_i(v) - \widehat{\chi}_i(v)\right)^2 = \sum_{i=1}^k \left\|f_i - \widehat{\chi}_i(v)\right\|^2 \le {\frac{k^4}{\Upsilon_G(k)}}.
\]
It follows that

\begin{equation}
\vol(B) \le  \frac{k^4}{\Upsilon_G(k)} \cdot \frac{\Upsilon_G(k) \beta \vol(V)}{C k^4\log{(1/\beta)}} = \frac{\beta \vol(V)}{C\log (1/\beta)}.
\end{equation}
By \eq{sep1}, and the definition of $B$, we have that for all $v \not\in B$, by setting $x^{(0)}=\chi_v$ we have that 
\begin{align*}
\left\| x^{(T)} - \frac{1}{\sqrt{\vol(S_j)}} {\chi}_{S_j}\right\| &= O\left(\frac{\log{n} }{ \Upsilon_G(k)\cdot \sqrt{\beta \vol(V)}} + \sqrt{\frac{k^4}{\Upsilon_G(k) \beta\vol(V)}} +\sqrt{\frac{Ck^4 \log{(1/\beta)}}{\Upsilon_G(k) \beta \vol(V)}}\right) \\
				&= O\left( \sqrt{\frac{k^4\log{(1/\beta)}}{\Upsilon_G(k) \beta \vol(V)}}\right),
\end{align*}
since it holds by \eq{gap} that  $\log{n}/\Upsilon=o(1)$. This implies that $V\setminus B\subseteq A$,
 yielding the claimed statement.
\end{proof}

So far we analysed the case for a single initial vector. To identify all the $k$ clusters simultaneously in $T$ rounds, we repeat this process multiple times with carefully chosen initial vectors. In particular, we need to ensure that we start the averaging process from at least one node in each cluster. This is the reason for us to introduce the \textsc{Seeding} step. By setting the probability $\bar{s}\cdot d_v/\vol(V)$ for every node $v$ to be active, it is easy to prove that with constant probability there is at least an active node in each cluster. 

Our analysis for the  \textsc{Query} step is based on the relation between our averaging procedure and  lazy random walks:
since any single random walk gets well mixed inside a cluster after $T$ steps, we expect that  the states of the nodes inside a cluster are similar.
 Namely, for the cluster $S_j$ to which the initial node of the $i$th vector belongs to, we expect that  
 $
x^{(T,i)}(v) \approx {1}/{\vol(S_j)}$
 for most $v\in S_j$ and $x^{(T,i)}(v) \approx 0$ otherwise. Hence, nodes from the same cluster will
choose the same label based on \eq{query}, while nodes from different clusters will have different labels. 


\begin{proof}[Proof of \thmref{cluster}]
For each node $v$, the probability that we start an averaging process with initial vector $\chi_v$ is equal to $ {\bar{s} \cdot d_v}/{\vol(V)}$, where $\bar{s} = \Theta\left(\frac{1}{\beta} \log\frac{1}{\beta}\right)$. Hence, the probability that there exists a $j$ such that no node from $S_j$ is chosen as initial node is at most
\[
\prod_{v \in S_j} \left(1 - \frac{\bar{s} \cdot d_v}{\vol(V)}\right) \le \prod_{v \in S_j} \mathrm{e}^{- {\bar{s} \cdot d_v}/{\vol(V)}} 
= \mathrm{e}^{-\bar{s}\sum_{v \in S_j} d_v / \vol(V)} \le \frac{1}{200 k},
\]
where we used the inequality $1-x \le \mathrm{e}^{-x}$ for $x \le 1$, the assumption on the size of the clusters, i.e., $\vol(S_j) \ge \beta \vol(V)$, and the trivial fact that $\beta \le 1/k$. As a consequence, with probability greater than $1/200$, for each cluster $S_j$, at least one node $v \in S_j$ is chosen as a starting node of the averaging process.

Next, we bound the probability that all the starting nodes belong to the set $A$ defined in \lemref{good}. By the algorithm description, the actual number of active nodes $s$ satisfies  $\Ex{s} = \bar{s}$. Therefore, it holds with probability $1-O(1)$ that  $s = O\left(\frac{1}{\beta} \log\frac{1}{\beta}\right)$. We assume that this event occurs in the rest of the proof. 
 Let $v_1,\dots,v_s$ be the starting nodes. By \lemref{good}, the probability that there exists a starting node $v_i$ not belonging to $A$ is at most
\begin{align*}
\Pro{\text{there is some starting node\ } v_i \not \in A} &\leq \frac{O(\bar{s}) \cdot (\vol(V)-\vol(A))}{\vol(V)}\le \frac{O(\bar{s})\cdot \beta}{C\log (1/\beta)} \le \frac{1}{200}.
\end{align*}
Hence, with  probability $1-O(1)$ every starting node belongs to $A$. For the rest of the proof we assume this is the case. For any node $v$, let $\mathcal{S}(v)$ be the cluster $v$ belongs to.
Then, by the definition of the set $A$, it holds  for any starting node $v_i$ that 
\begin{equation}
\label{eq:closenorm}
\left\| x^{(T,i)} - \frac{{\chi}_{\mathcal{S}(v_i)}}{\sqrt{\vol(\mathcal{S}(v_i))}} \right\| = O\left(\sqrt{\frac{k^4 \log{(1/\beta)}}{\Upsilon_G(k)\beta \vol(V)}}\right).
\end{equation}
 Observe that a node $v$ is misclassified by the algorithm only if there exists $i \in \{1,\dots,s\}$ such that
\begin{equation}
\label{eq:mis}
 \left|\frac{x^{(T,i)}(v)}{\sqrt{d_v }}-\frac{\chi_{\mathcal{S}(v_i)}(v)}{\sqrt{d_v \cdot \vol(\mathcal{S}(v_i))}}\right|^2 >  {\frac{1}{4\beta^2 \vol(V)^2}}.
\end{equation}
Then, by \eq{closenorm} the total volume of misclassified nodes is at most
\begin{align*}
\lefteqn{O(\bar{s}) \cdot  \left\| x^{(T,i)} - \frac{1}{\sqrt{\vol(\mathcal{S}(v_i))}} {\chi}_{\mathcal{S}(v_i)}\right\|^2 \cdot 4\beta^2 \vol(V)^2 }\\
	&= O\left(\frac{1}{\beta} \log\frac{1}{\beta} \cdot \frac{k^4}{\Upsilon_G(k)} \cdot  {\frac{  \log{(1/\beta)}}{ \beta \vol(V)}} \cdot \beta^2 \vol(V)^2\right) \\
	&= O\left( \frac{k^4}{\Upsilon_G(k)} \log^2\frac{1}{\beta} \right) \vol(V).
\end{align*}
Combining this with the assumption \eq{gap} proves the first statement.
The second statement follows by the fact that the total communication among all nodes in each round is  $O(m \cdot  (1/\beta)
 \log(1/\beta))$ words.
\end{proof}

\section{Experiments}
\label{sec:exp}

Now we present experimental results for our sparsification algorithm on both synthetic and real-world datasets. To report a detailed and quantitive  result, we will compare the clustering results of the following two approaches: (1) apply spectral clustering on the original input dataset; (2)  apply spectral clustering on the graph returned by our sparsification algorithm.

Besides giving the visualised results of our algorithm on various datasets, we use two functions to measure the quality of the above-mentioned two approaches:
(1)  For the synthetic datasets for which the underlying ground-truth clustering is known, 
 the quality of a clustering algorithm is measured  by the ratio of misclassified points, i.e., 
 \[
 \mathsf{err}(A_1,\ldots, A_k)\triangleq \frac{1}{n}\cdot \sum_{i=1}^k | \{v\in A_i: v\not\in S_i \} |,
 \]
where $\{S_1,\cdots, S_k\}$ is the underlying ground-truth clustering and $\{A_1,\dots,A_k\}$ is the one returned by the clustering algorithm.
(2) 
 For datasets for which a ground-truth clustering is not well-defined, 
 the quality of a clustering is measured by  the \emph{normalised cut value}  defined by $$
\mathsf{ncut}(A_1,\ldots, A_k)  \triangleq \sum_{i=1}^k \frac{w(A_i, V\setminus A_i)}{\vol (A_i)},$$
which is a standard objective function to be minimised for spectral clustering algorithms \cite{shiMalikNC,vonLuxSC}.  All the experiments are conducted with Matlab and we use an implementation of the classical spectral clustering algorithm described in~\cite{ng2001spectral}.

\subsection{Datasets} We test the algorithms in the following three  synthetic and real-world datasets, which are visualised in \figref{data}. 

\begin{itemize}
\item \texttt{Twomoons}: this dataset consists of $n$ points in $\mathbb{R}^2$, where $n$ is chosen between $1,000$ and $15,000$. 
We consider each point to be a node. For any two nodes $u,v$, we add an edge with weight $w(u,v) = \exp(- \| u- v \|^2/2\sigma^2)$, where $\sigma=0.1$. 
\item \texttt{Gaussians}: this dataset consists of $n$ points in $\R^2$, where $n$ is chosen between $1,000$ and $15,000$.  Each point is sampled from a uniform mixture of 3 isotropic Gaussians of variance $0.04$. The similarity graph is constructed in the same way as \texttt{Twomoons}, and we set $\sigma=1$ here.
\item \texttt{Sculpture}: we use a $73\times 160$ version of a photo of \emph{The Greek Slave}\footnote{\href{http://artgallery.yale.edu/collections/objects/14794}{http://artgallery.yale.edu/collections/objects/14794}} where each pixel is viewed as a node. To construct a similarity graph, we map each pixel to a point in $\R^5$, i.e., $(x,y,r,g,b)$, where the last three coordinates are the RGB values. For any two nodes $u,v$, we put an edge between $u$ and $v$ with weight $w(u,v) = \exp(- \|u-v\|^2/2\sigma^2)$, where $\sigma=20$. This results in a graph with about $11,000$ nodes and $k=3$ clusters.
\end{itemize} 
These datasets are essentially the ones used in \cite{CSWZ16}, which  studies the effects of spectral sparsification on  clustering. This   makes it possible to easily compare our results with the state-of-the-art. The choice of $\sigma$ varies for different datasets, since
 they have in general different intra-cluster variance. 
 There are several heuristics to choose the ``correct'' value of $\sigma$ (see, e.g., the classical reference~\cite{ng2001spectral}). In our case the value of $\sigma$ is chosen so that the spectral gap of the original similarity graph is large. This ensures that the clusters in the graph are well-defined, and spectral clustering outputs a meaningful clustering.

\begin{figure}[h]
    \centering
    \begin{subfigure}[b]{0.25\textwidth}
        \includegraphics[width=\textwidth]{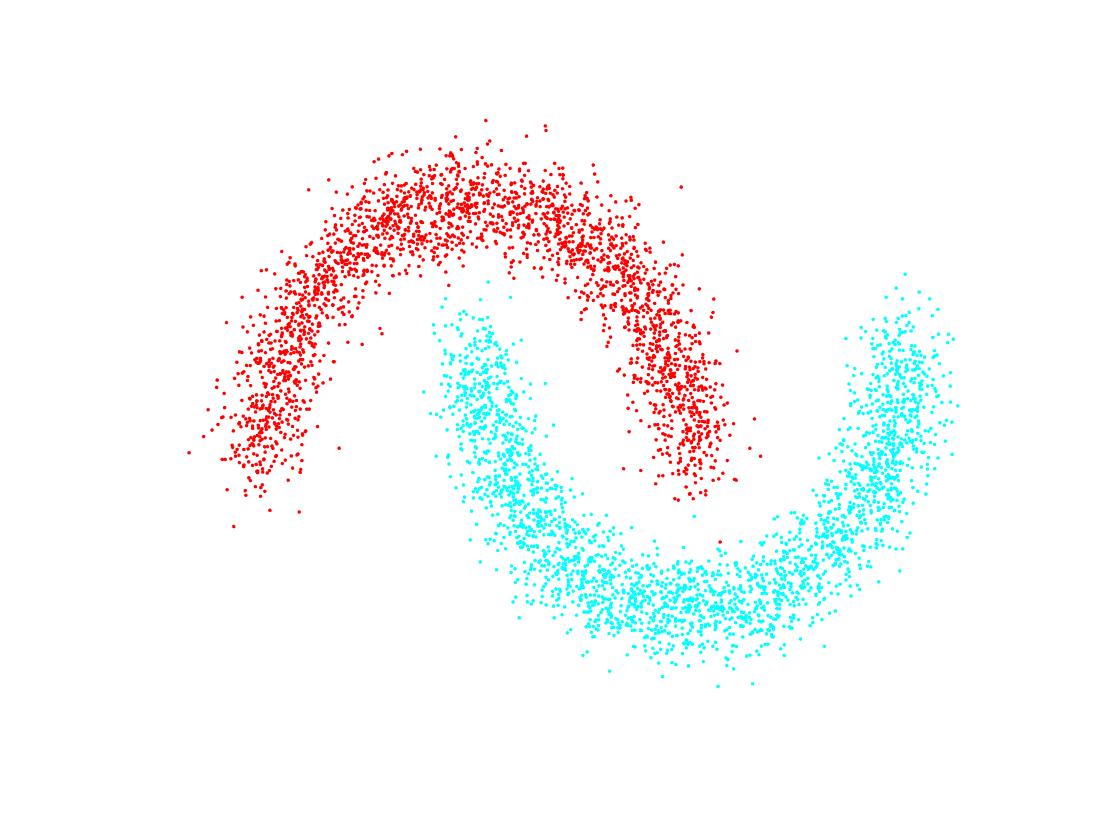}
        \caption{\texttt{Twomoons}}
        \label{fig:moons}
    \end{subfigure}
    \qquad\qquad
     \begin{subfigure}[b]{0.25\textwidth}
        \includegraphics[width=\textwidth]{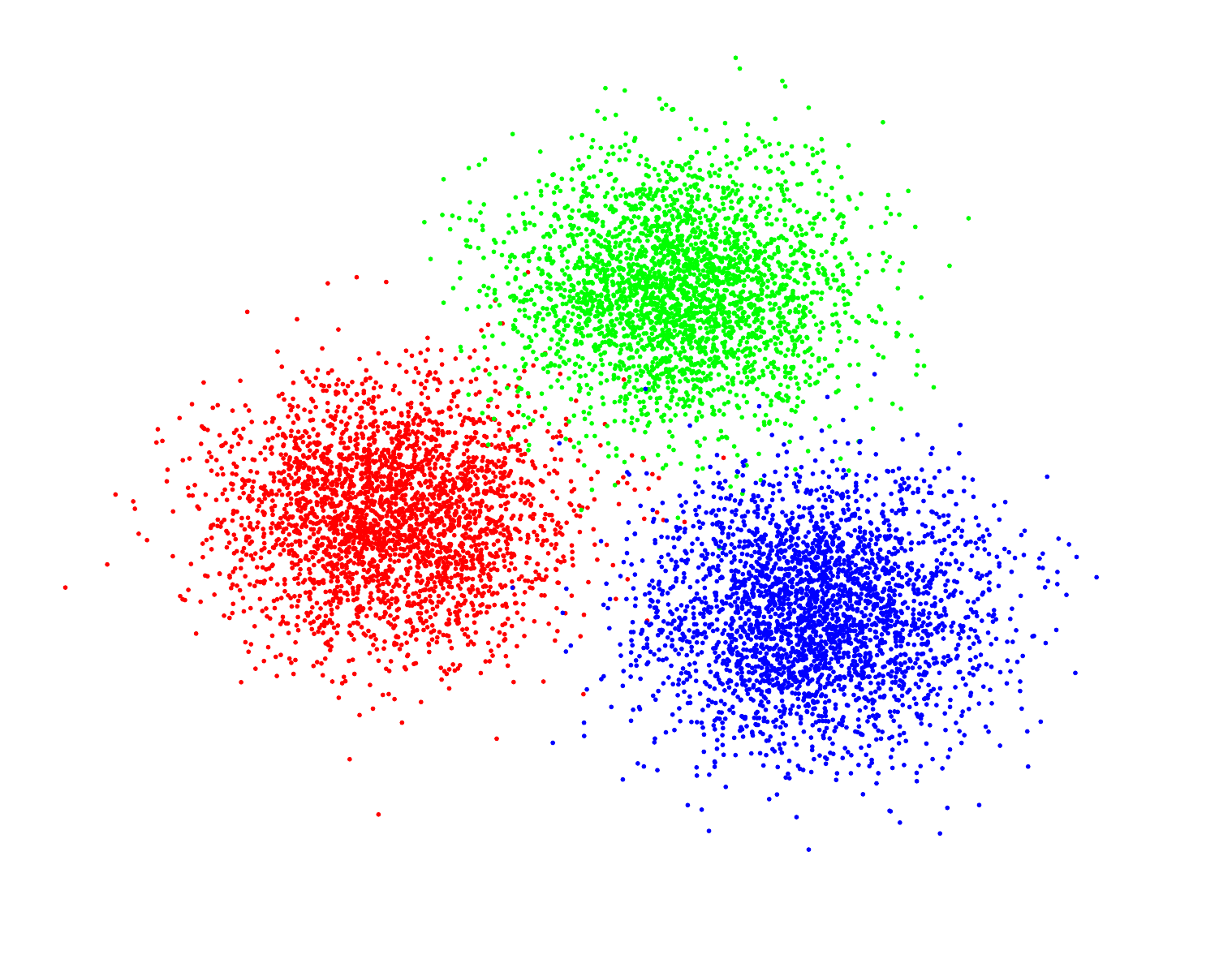}
        \caption{\texttt{Gaussians}}
        \label{fig:gauss}
    \end{subfigure}
        \qquad\qquad
    \begin{subfigure}[b]{0.25\textwidth}\centering
        \includegraphics[width=0.4\textwidth]{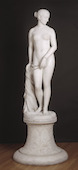}
        \caption{\texttt{Sculpture}}
        \label{fig:scu}
    \end{subfigure}
    \caption{Visualisation of the  datasets used in our experiments.}\label{fig:data}
\end{figure}

\subsection{Results on Clustering Quality}  We test the performance of our algorithm on the three datasets. Notice that the sampling probability 
of the edges in our sparsification algorithm involves the factor $\tau \ge C/\lambda_{k+1}$. To find a desired value of $\tau$, we use the following doubling method: starting with $\tau=0.1$, we double the value of $\tau$ each time, until the spectral gap $|\lambda_{k+1}-\lambda_k|$ of the resulting matrices doesn't  change significantly. Remarkably, for all the datasets considered in the paper, $\tau=1.6$ always suffices for our purposes. Notice that this method will only increase the time complexity of our algorithm by at most a poly-logarithmic factor of $n$.

For the \texttt{Twomoons} and \texttt{Gaussians} datasets, for all the tested graphs with size ranging from $1,000$ to $15,000$ points, our sparsified graphs require only about $0.14\% \sim 3.13\%$ of the total edges. The error ratios of spectral clustering on the original datasets and our sparsified graphs are listed respectively as $\mathsf{err}1$ and $\mathsf{err}2$, and are always very close. See Table~\ref{twomoontable} and Table~\ref{gausstable}  for details.

\begin{table}[t]
\caption{Experimental results for the \texttt{Twomoons} dataset, where $\tau=0.8$.}
\label{twomoontable}
\vskip 0.15in
\centering
\begin{tabular}{lllll}
\toprule
  $n$ & $\#$ edges~($\%$) & $\mathsf{err}1$~($\%$) & $\mathsf{err}2$~($\%$)\\
\midrule
$1,000$   & $1.56$ & $0.900$ & $0.600$ \\
$2,000$ & $0.86$ &  $0.150$ & $0.100$ \\
$4,000$    & $0.48$ & $0.150$ & $0.175$ \\
$8,000$   & $0.26$ & $0.086$ & $0.088$         \\
$10,000$     & $0.22$ &  $0.120$ & $0.150$ \\
$15,000$      & $0.14$ &  $0.080$ & $0.100$ \\
\bottomrule
\end{tabular}
\end{table}

\begin{table}[t]
\caption{Experimental results for the \texttt{Gaussians} dataset, where $\tau=1.6$.}
\label{gausstable}
\centering
\begin{tabular}{lllll}
\toprule
  $n$ & $\#$ edges~($\%$) & $\mathsf{err1}$~($\%$) & $\mathsf{err2}$~($\%$) \\
\midrule
$1,000$   & $3.13$ & $1.700$ & $1.900$ \\
$2,000$ & $1.75$ &  $1.350$ & $1.650$ \\
$4,000$    & $0.96$ & $0.400$ & $0.417$ \\
$8,000$   & $0.66$ & $0.128$ & $0.140$         \\
$10,000$     & $0.42$ &  $0.113$ & $0.119$ \\
$15,000$      & $0.29$ &  $0.125$ & $0.148$ \\
\bottomrule
\end{tabular}
\end{table}

The \texttt{Sculpture} dataset corresponds to a similarity graph of $n=11,680$ nodes and $68$ million  edges.   We run  spectral clustering on both the input graph and  our sparsified one, 
 and compute the normalised cut values of each clustering in the original input graph. By  setting $\tau=1.6$, our algorithm samples only  $0.37\%$ of the edges~($320,000$) from the input graph.
 The normalised cut value of spectral clustering on the original dataset is $0.0938$, while the normalised cut value of spectral clustering on our sparsified graph is $0.0935$. The visualisation of the two clustering results are almost identical, as shown in Figure~\ref{fig:rescu}.

\begin{figure*}[htb]
   \centering
    \begin{subfigure}[b]{0.20\textwidth}\centering
        \includegraphics[width=0.45\textwidth]{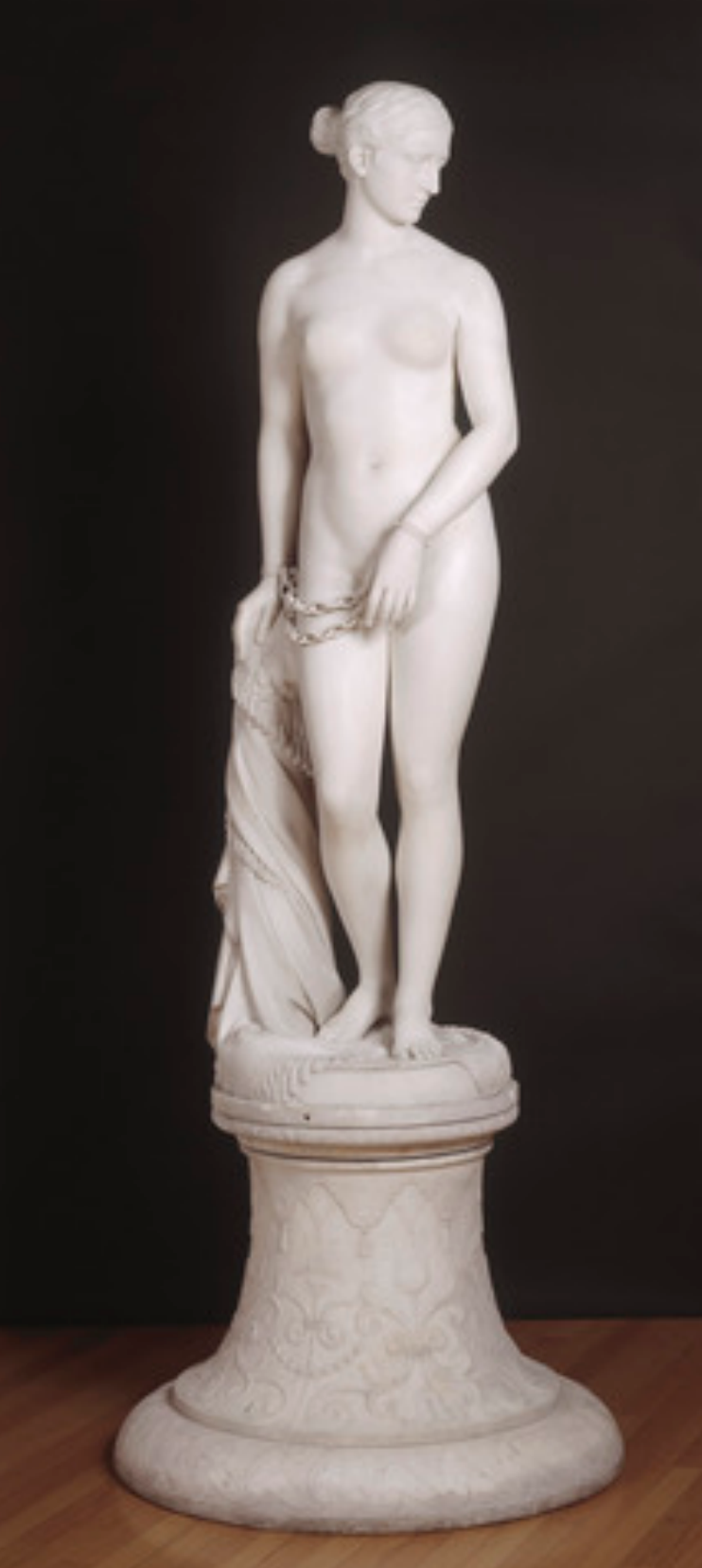}
        \label{fig:scu}
    \end{subfigure}
    \qquad\qquad
    \begin{subfigure}[b]{0.20\textwidth}\centering
        \includegraphics[width=0.45\textwidth]{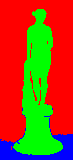}
        \label{fig:scu2}
    \end{subfigure}
    \qquad\qquad
     \begin{subfigure}[b]{0.2\textwidth}\centering
        \includegraphics[width=0.45\textwidth]{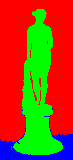}
        \label{fig:scu2}
    \end{subfigure}
\caption{Visualisation of the results on \texttt{Sculpture}. The left-side picture is the original input dataset, the middle one is the outcome of spectral clustering on the original input dataset, while the right-side picture is the outcome of spectral clustering on our sparsified graph.}\label{fig:rescu}
\end{figure*}


\textbf{Acknowledgement:}
We would like to thank Dr. Emanuele Natale and Prof. Luca Trevisan who found a gap in the analysis of the algorithm in an earlier version of our paper. We present an improved  algorithm to avoid the issues occurring in the previous version of our paper. 
Moreover, compared with the previous one, our improved algorithm  does not require the input graph to be regular.

\bibliographystyle{alpha}

\bibliography{references,ref} 

\end{document}